\documentclass[a4paper,11pt]{article}

\usepackage{vmargin}

\setmarginsrb{3.2cm}{1.6cm}{3.2cm}{3.cm}{.7cm}{1.2cm}{0.8cm}{1.4cm}

\usepackage{amssymb,amstext}
\usepackage{cite,scrtime}
\usepackage{graphics}
\usepackage{amsmath}
\usepackage{mathrsfs}
\usepackage{fancyhdr}
\usepackage{verbatim}
\usepackage{boxedminipage}
\usepackage{colortbl}
\usepackage{xspace}

\newtheorem{fait} {Fact}
\newtheorem{lem}  {Lemma}
\newtheorem{prop} {Proposition}

\newtheorem{rgl}  {Rule}
\newtheorem{defi} {Definition}
\newtheorem{claimN} {Claim}

\newcommand{\rlem}   [1] {Lemma~\ref{#1}\xspace}

\newcommand{\rrgl}   [1] {Rule~\ref{#1}\xspace}

\newcommand{\domrb}  [0] {\textsc{Red-Blue Dominating Set}\xspace}
\newcommand{\dom}    [0] {\textsc{Dominating Set}\xspace}

\newcommand{\drb}    [0] {rbds\xspace}
\newcommand{\rbds}{rbds\xspace}
\newcommand{\RBDS}{\textsc{RBDS}\xspace}

\newcommand{\YES}{\textsc{Yes}}

\newcommand{\NP}{\ensuremath{\text{NP}}\xspace}

\newtheorem{theorem}{Theorem}

\newenvironment{proof}{\noindent \textit{Proof. }}{\hfill$\square$\vspace{.2cm}}
\newenvironment{proofof}{\noindent \textit{Proof of Theorem~\ref{th:main}. }}{\hfill$\square$\vspace{.2cm}}

  \renewenvironment{thebibliography}[1]{%
    \begin{oldthebibliography}{#1}%
      \setlength{\parskip}{0.4ex}%
      \setlength{\itemsep}{0.4ex}%
  }%
  {%
    \end{oldthebibliography}%
  }

\usepackage{todonotes}
\newcommand{\add}   [1] {\textcolor{red} {#1}}
\newcommand{\addOK}   [1] {\textcolor{black} {#1}}

\newcommand{\modif} [2] {\add{#2}}
\newcommand{\modifOK} [2] {\addOK{#2}}

\newcommand{\ig}[1]{\textcolor{blue}{\fbox{\fbox{\textcolor{blue}{#1}}}}}

\title{A Linear Kernel for Planar Red-Blue Dominating Set\thanks{A preliminary short version of this work appeared in the\emph{ Proceedings of the 12th Cologne-Twente Workshop on Graphs and Combinatorial Optimization, pages 117-120, Enschede, Netherlands, May 21-23,  2013}.}}

\author{Valentin Garnero\footnote{Universit\'e de Montpellier,  LIRMM, Montpellier, France.
        E-mail: \texttt{\small{valentin.garnero@lirmm.fr}}.}$\
        $,$\  $
        Ignasi Sau\footnote{CNRS, LIRMM--Universit\'e de Montpellier, Montpellier, France.
        E-mail: \texttt{\small{ignasi.sau@lirmm.fr}}.}$\
        $, and 
        Dimitrios M. Thilikos\footnote{CNRS, LIRMM, Montpellier, France. Department of Mathematics, National \& Kapodistrian University of Athens, Greece.
        E-mail: \texttt{\small{sedthilk@thilikos.info}}.}\ \footnote{Co-financed by the European Union (European Social Fund - ESF) and Greek national funds through the Operational Program ``Education and Lifelong Learning'' of the National Strategic Reference Framework (NSRF) - Research Funding Program: ``Thales. Investing in knowledge society through the European Social Fund''.}
        }

\date{}

\begin{document}

 \maketitle


 \begin{abstract}
In the \textsc{Red-Blue Dominating Set} problem, we are given a bipartite graph
$G = (V_B \cup V_R,E)$ and an integer $k$, and asked whether $G$ has a subset $D \subseteq V_B$ of at most $k$ ``blue'' vertices such that each ``red'' vertex from $V_R$ is adjacent to a vertex in $D$. We provide the first explicit linear kernel for this problem on planar graphs, of size at most $43k$.

\vspace{0.25cm}
\end{abstract}

\section{Introduction}
\label{sec:intro}


\paragraph{Motivation.} The field of parameterized complexity (see~\cite{DF99,FG06,Nie06}) deals with algorithms for decision problems whose instances consist of a pair $(x,k)$, where~$k$ is known as the \emph{parameter}. A fundamental concept in this area is that of \emph{kernelization}. A kernelization algorithm, or \emph{kernel}, for a parameterized problem takes an instance~$(x,k)$ of the problem and, in time polynomial in $|x| + k$, outputs an equivalent instance~$(x',k')$ such that $|x'|, k' \leq g(k)$ for some
function~$g$. The function~$g$ is called the \emph{size} of the kernel and may
be viewed as a measure of the ``compressibility'' of a problem using
polynomial-time preprocessing rules. A natural problem in this context is to find
polynomial or linear kernels for problems that admit such kernelization algorithms.

A celebrated result in this area is the linear kernel for \textsc{Dominating Set} on planar graphs by Alber \emph{et al}.~\cite{AFN04}, which gave rise to an explosion of (meta-)results on linear kernels on planar graphs~\cite{GuNi07} and other sparse graph classes~\cite{BFL+09,FLST10,KLP+12}.
Although of great theoretical importance, these meta-theorems have two important drawbacks from a practical point of view. On the one hand, these results rely on a problem property called \emph{Finite Integer Index}, which guarantees the {\sl existence} of a linear kernel, but nowadays it is still not clear how and when such a kernel can be effectively {\sl constructed}. On the other hand, at the price of generality one cannot hope that general results of this type may directly provide explicit reduction rules and small constants for particular graph problems. Summarizing, as mentioned explicitly by Bodlaender \emph{et al}.~\cite{BFL+09}, these meta-theorems provide simple criteria to decide whether a problem admits a linear kernel on a graph class, but finding linear kernels with reasonably small constant factors for concrete problems remains a worthy investigation topic.

\paragraph{Our result.} In this article we follow this research avenue and focus on the \textsc{Red-Blue Dominating Set} problem (\RBDS for short) on planar graphs. In the \textsc{Red-Blue Dominating Set} problem, we are given a bipartite\footnote{In fact, this assumption is not necessary, as if the input graph $G$ is not bipartite, we can safely remove all edges between vertices of the same color.} graph $G = (V_B \cup V_R,E)$ and an integer $k$, and asked whether $G$ has a subset $D \subseteq V_B$ of at most $k$ ``blue'' vertices such that each ``red'' vertex from $V_R$ is adjacent to a vertex in $D$. This problem appeared in the context of the European railroad network~\cite{Wei98}.
From a (classical) complexity point of view, finding a red-blue dominating set (or \drb for short) of minimum size is \NP-hard on planar graphs~\cite{ABFN00}. From a parameterized complexity perspective, \RBDS parameterized by the size of the solution is $W[2]$-complete on general graphs and {\sf FPT} on planar graphs~\cite{DF99}. It is worth mentioning that \RBDS plays an important role in the theory of non-existence of polynomial kernels for parameterized problems~\cite{DLS09}.


The fact that \RBDS involves a {\sl coloring} of the vertices of the input graph makes it unclear how to make the problem fit into the general frameworks of~\cite{GuNi07,BFL+09,FLST10,KLP+12}. In this article we provide the first explicit (and quite simple) polynomial-time data reduction rules for \textsc{Red-Blue Dominating Set} on planar graphs, which lead to a linear kernel for the problem.

\begin{theorem}
\label{th:main}
\textsc{Red-Blue Dominating Set} parameterized by the solution size has a linear kernel on planar graphs. More precisely, there exists a polynomial-time algorithm that for each planar instance $(G,k)$, either correctly reports that $(G,k)$ is a \textsc{No}-instance, or returns an equivalence instance $(G',k')$ such that $k'\leq k$ and $|V(G')| \leq 43 \cdot k'$.
\end{theorem}

This result complements several explicit linear kernels on planar graphs for other domination problems such as \textsc{Dominating Set}~\cite{AFN04}, \textsc{Edge Dominating Set}~\cite{WYGC13,GuNi07}, \textsc{Efficient Dominating Set}~\cite{GuNi07}, \textsc{Connected Dominating Set}~\cite{LMS11,GuIm10}, or \textsc{Total Dominating Set}~\cite{GaSa12}. It is worth mentioning that our constant is considerably smaller than most of the constants provided by these results. Since one can easily reduce the \textsc{Face Cover} problem on a planar graph to \RBDS (without changing the parameter)\footnote{Just consider the \emph{radial graph} corresponding to the input graph $G$ and its dual $G^*$, and color the vertices of $G$ (resp. $G^*$) as red (resp. blue). }, the result of Theorem~\ref{th:main} also provides a linear \emph{bikernel} for \textsc{Face Cover} (i.e., a polynomial-time algorithm that given an input
of \textsc{Face Cover}, outputs an equivalent instance of \RBDS\ with a graph whose size is linear in $k$). To the best of our knowledge, the best existing kernel for \textsc{Face Cover} is quadratic~\cite{KLL02}. Our techniques are much inspired by those of Alber~\emph{et al}.~\cite{AFN04} for \textsc{Dominating Set}, although our reduction rules and analysis are slightly simpler. We start by describing in Section~\ref{sec:RedRules} our reduction rules for \textsc{Red-Blue Dominating Set} when the input graph is embedded in the plane, and in Section~\ref{sec:analysis} we prove that the size of a reduced plane \YES-instance is linear in the size of the desired red-blue dominating set, thus proving Theorem~\ref{th:main}. Finally, we conclude with some directions for further research in Section~\ref{sec:concl}.

\section{Reduction rules} \label{sec:RedRules}

In this section we propose reduction rules for \domrb, which are largely inspired by the rules that yielded the first linear kernel for \dom on planar graphs~\cite{AFN04}. The idea is to either replace the neighborhood of some blue vertices by appropriate gadgets, or to remove some blue vertices and their neighborhood when we can assume that these blue vertices belong to the dominating set. We would like to point out that our rules have also some points in common with the ones for the current best kernel for \dom~\cite{CFKX07}. In Subsection~\ref{elementary} we present two easy elementary rules that turn out to be helpful in simplifying the instance, and then in Subsections~\ref{Rsom} and~\ref{Rpair} we present the rules for a single vertex and a pair of vertices, respectively.



Before starting with the reduction rules, we need a definition.

\begin{defi}\label{def:reduce}
We say that a graph $G$ is \emph{reduced under a set of rules} if either none of these rules can by applied to $G$, or the application of any of them creates a graph isomorphic to $G$.
\end{defi}

With slight abuse of notation, we simply say that a graph is \emph{reduced} if it reduced under the whole set of reduction rules that we will define, namely Rules 1, 2, 3, and 4.

We would like to point out that the above definition differs from the usual definition of \emph{reduced} graph in the literature, which states that a graph is reduced if the corresponding reduction rules cannot be applied anymore. We diverge from this definition because, for convenience, we will define reduction rules  that could be applied {\sl ad infinitum} to the input graph, such as Case~2 of Rule~\ref{rgl_pair} defined in Subsection~\ref{Rpair}. For algorithmic purposes, the reduction rules that we will define are all local and concern the neighborhood of at most 2 vertices, which is replaced with gadgets of constant size. Therefore, in order to know when a graph is reduced (see Definition~\ref{def:reduce}), the fact whether the original and the modified graph are isomorphic or not  can be easily checked locally in constant time.

\subsection{Elementary rules} \label{elementary}

The following two elementary rules enable us to simplify an instance of \RBDS. We would like to point out that similar rules have been provided by Weihe~\cite{Wei98} in a more applied setting. We first need the definition of neighborhood.

\begin{defi}
Let $G=(V_B\cup V_R,E)$ be a graph.
The \emph{neighborhood} of a vertex $v \in V_B \cup V_R$ is the set $N(v) = \{u : \{v,u\} \in E \}$.
The \emph{neighborhood} of a pair of vertices $v,w \in V_B$ is the set $N(v,w) = N(v) \cup N(w)$.
\end{defi}


\begin{rgl}     \label{rgl_bleu}
Remove any blue vertex $b$ such that $N(b) \subseteq N(b')$ for some other blue vertex $b'$.
\end{rgl}

\begin{rgl}    \label{rgl_rouge}
Remove any red vertex $r$ such that $N(r) \supseteq N(r')$ for some other red vertex $r'$.
\end{rgl}

\begin{lem} \label{lem_corr_elem}
Let $G=(V_B\cup V_R,E)$ be a graph.
If $G'$ is the graph obtained from $G$ by the application of Rule \ref{rgl_bleu} or \ref{rgl_rouge}, then there is a \drb in $G$ of size at most $k$ if and only if there is one in $G'$.
\end{lem}

\begin{proof}
For \rrgl{rgl_bleu}, if $N(b) \subseteq N(b')$ for two blue vertices $b$ and $b'$, then any solution containing $b$ can be transformed to a solution containing $b'$ in which the set of dominated red vertices may have only increased. For \rrgl{rgl_rouge}, 
if $N(r') \subseteq N(r)$ for two red vertices $r$ and $r'$, then any blue vertex dominating $r'$ dominates also $r$.
\end{proof} 


\subsection{Rule for a single vertex} \label{Rsom}

In this subsection we present a rule for removing a blue vertex when it is necessarily a dominating vertex. For this we first need the definition of private neighborhood.

\begin{defi}
Let $G=(V_B\cup V_R,E)$ be a graph.
The \emph{private neighborhood} of a blue vertex $b$ is the set $P(b) = \{r \in N(b) : N(N(r)) \subseteq N(b) \}$.
\end{defi}

Let us remark that for (classical) \dom, each neighborhood is split into three subsets~\cite{AFN04}. The third one corresponds to our private neighborhood, but since non-private neighbors can be used to dominate the private ones, an intermediary set is necessary for (classical) \dom. In our problem this does not occur because non-private vertices are red and thus cannot belong to a \drb. This is one of the reasons why our rules are simpler.

\begin{rgl}  \label{rgl_som}
Let $v \in V_B$ be a blue vertex. If $|P(v)| \geq 1$:
\begin{itemize}\itemsep0em
\item remove $v$ and $N(v)$ from $G$,
\item decrease the parameter $k$ by $1$.
\end{itemize}

\end{rgl}

Our \rrgl{rgl_som} corresponds to Rule 1 for (classical) \dom~\cite{AFN04}. In both rules we can safely assume that vertex $v$ belongs to the dominating set, but for \RBDS we can remove it together with its neighborhood (and decrease the parameter accordingly); this is not possible for \dom, since vertices of the neighborhood possibly belong to the dominating set as well, hence they cannot be removed and a gadget is added to enforce $v$ to be dominating. We prove in the following lemma that \rrgl{rgl_som} is safe.

\begin{lem} \label{lem_corr_1}
Let $G=(V_B\cup V_R,E)$ be a graph reduced under Rules~\ref{rgl_bleu} and~\ref{rgl_rouge} and let $v \in V_B$.
If $(G',k-1)$ is the instance obtained from $(G,k)$ by the application of \rrgl{rgl_som} on a vertex $v$,
then there is a \drb in $G$ of size at most $k$ if and only if there is one in $G'$ of size at most $k-1$.
\end{lem}

\begin{proof}
Let $D$ be a \drb in $G$ with $|D| \leq k$. Since $G$ is reduced under Rules~\ref{rgl_bleu}  and~\ref{rgl_rouge}, and \rrgl{rgl_som} can be applied on vertex $v$, necessarily $v \in D$ in order to dominate the vertices in $P(v)$. Since $G'$ does not contain any vertex of $N(v)$, $D\setminus\{v\}$ is a \drb of $G'$ of size at most $k-1$.
Conversely, let $D'$ be a \drb in $G'$ with $|D'| \leq k'$. Clearly $D' \cup \{v\}$ is a \drb of $G$ of size at most $k'+1$.
\end{proof}

In the following fact we prove that if we assume that Rules~\ref{rgl_bleu} and~\ref{rgl_rouge} have been exhaustively applied, then \rrgl{rgl_som} is equivalent to a simpler rule that consists in removing an appropriate connected component of size two and decreasing the parameter by one.

\begin{fait}
Let  $G=(V_B\cup V_R,E)$ be a graph reduced under Rules~\ref{rgl_bleu} and~\ref{rgl_rouge}, and let $v \in V_B$. Then $P(v) \neq \emptyset$ if and only if $N(v)=\{r\}$ with $N(r)= \{v\}$ for some $r \in V_R$.
\end{fait}

\begin{proof}
First, if $N(v)=\{r\}$ with $N(r)= \{v\}$ then $P(v)=\{r\} \neq \emptyset$. Conversely, suppose that $P(v) \neq \emptyset$, let $r \in P(v)$,  and assume for contradiction that  $|N(v)| \geq 2$ or $N(r) \setminus \{v\} \neq \emptyset$. In the former case,  since by Rule~\ref{rgl_rouge} the neighborhood of vertex $r$ is incomparable  with that of other red vertices in $N(v)$, necessarily $\{v,b\} \subseteq N(r)$ for some $b \in V_B$. And since vertex $r$ is a private neighbor of $v$, necessarily $N(b) \subseteq N(v)$, contradicting the hypothesis that $G$ is reduced under Rule~\ref{rgl_bleu}. In the latter case, it also follows that $\{v,b\} \subseteq N(r)$ for some $b \in V_B$, and we reach the same contradiction. \end{proof}

\subsection{Rule for a pair of vertices} \label{Rpair}

We now provide a rule for either reducing the size of the neighborhood of a pair of blue vertices, or for removing some blue vertices together with their neighborhood. For this, we first define the private neighborhood of a pair of blue vertices.

\begin{defi}
Let $G=(V_B\cup V_R,E)$ be a graph. The \emph{private neighborhood} of a pair of blue vertices $v,w \in V_B$ is the set $P(v,w) = \{r \in N(w,v) : N(N(r)) \subseteq N(v,w) \}$.
\end{defi}

We would like to note that the definition of private neighborhood is similar to that of the third subset of neighbors defined for (classical) \dom \cite{AFN04}.

\begin{rgl}  \label{rgl_pair}            
Let $v,w$ be two distinct blue vertices such that $|P(v,w)| \geq 1$.
Let $\mathcal{D} = \{d \in V_B: P(v,w) \subseteq N(d)  \}$.
We distinguish the following cases:
\begin{enumerate}

\item $P(v,w) \nsubseteq N(v)$ and $P(v,w) \nsubseteq N(w)$:

\begin{itemize}\itemsep0em
\item remove $P(v,w)$ from $G$,
\item add two new red vertices $r',r''$ and the edges $\{v,r' \},\{w,r'' \}$,
\item for each vertex $d \in \mathcal{D}$, add the edges $\{d,r' \},\{d,r'' \} $.
\end{itemize}


\item $P(v,w) \subseteq N(v)$ and $P(v,w) \subseteq N(w)$:

\begin{itemize}\itemsep0em
\item remove $P(v,w)$ from $G$,
\item add a new red vertex $r$ and the edges $\{v,r \},\{w,r \}$,
\item for each vertex {$d \in \mathcal{D}$}, add the edge $\{d,r \}$.
\end{itemize}

\item $P(v,w) \subseteq N(v)$ and $P(v,w) \nsubseteq N(w)$:

\begin{itemize}\itemsep0em
\item remove $P(v,w)$ from $G$,
\item add a new red vertex $r'$ and the edge $\{v,r' \}$,
\item for each vertex {$d \in \mathcal{D}$}, add the edge $\{d,r' \}$.

\end{itemize}

\item $P(v,w) \nsubseteq N(v)$ and $P(v,w) \subseteq N(w)$:

\begin{itemize}
\item symmetrically to Case 3.
\end{itemize}

\end{enumerate}
\end{rgl}


Again, our \rrgl{rgl_pair} corresponds to Rule 2 for (classical) \dom~\cite{AFN04}. Remark that, if $\mathcal{D}$ is empty, then the added vertices $r'$ and $r''$ have degree one, and hence \rrgl{rgl_som} can be applied to remove vertex $v$ or $w$. Observe also that, if $P(v,w) \not \subseteq N(w)$ (or  $N(w)$), then there is a red vertex in $ P(v,w) \setminus N(w)$, which plays a role similar to the added vertex $r'$; proving this observation is the key point to prove that \rrgl{rgl_pair} is safe. 
\begin{lem} \label{lem_corr_2}
Let $G=(V_B\cup V_R,E)$ be a graph reduced under Rules \ref{rgl_bleu} and \ref{rgl_rouge} and let $v,w$ be two distinct blue vertices.
If $(G',k)$ is the instance obtained from $(G,k)$ by the application of \rrgl{rgl_pair} on $v$ and $w$,
then there is a \drb in $G$ of size at most $k$ if and only if there is one in $G'$ of size at most $k$.
\end{lem}

\begin{proof}
We distinguish the four possible cases of application of \rrgl{rgl_pair}. For each case we prove that $G$ has a solution of size $k$ if and only if $G'$ has one.

\begin{enumerate}\itemsep0em
\item
Let $D$ be a \drb in $G$ of size $k$.
Since we are in Case~1, if $|D\cap N(P(v,w))| > 2$ then we can assume that $v,w \in D$ because $N(v,w)$ is a superset of the neighborhood of any pair of vertices in $N(P(v,w))$.
Therefore, in $G'$, the vertices $r',r''$ are dominated either by $v,w$ or by some vertex $d \in \mathcal{D}$ such that, in the graph $G$,  $P(v,w) \subseteq N(d)$. Hence $D$ is a \drb in $G'$ of size at most $k$.

Conversely, let $D'$ be a \drb in $G$ of size $k$.
Since $r'$ and $r''$ need to be dominated, we have that either $v,w \in D'$ or $d \in D'$ for some $d  \in mathcal{D}$. Hence $D'$ is a \drb in $G$ of size at most $k'$.

Observe that, since we are in Case~1, in $G$ there is a vertex in $P(v,w) \setminus N(w)$ (that is, a private neighbor which cannot be dominated by $w$); and similarly, there is a vertex in  $P(v,w) \setminus N(v)$. Hence $|P(v,w)| \geq 2$, and therefore the rule does not increase the number of vertices of the graph.

\item
Let $D$ be a \drb in $G$ of size $k$.
Since we are in Case~2, vertex $r$ is dominated by some vertex $d \in \mathcal{D} \cup \{v,w\}$.
Hence $D$ is a \drb in $G'$ of size at most $k$.

Conversely, let $D'$ be a \drb in $G$ of size $k$. Since $r$ needs to be dominated, we have that $d \in D'$ for some $d  \in \mathcal{D} \cup \{v,w\}$. Hence $D'$ is a \drb in $G$ of size at most $k$.

\item
Let $D$ be a \drb in $G$ of size $k$.
Since we are in Case~2,  vertex $r'$ is dominated by some vertex $d \in \mathcal{D} \cup \{v\}$.
Hence $D$ is a \drb in $G'$ of size at most $k$.

Conversely, let $D'$ be a \drb in $G$ of size $k$. Since $r'$ needs to be dominated, we have that $d \in D'$ for some $d  \in \mathcal{D} \cup \{v\}$. Hence $D'$ is a \drb in $G$ of size at most $k$.

\item Symmetrically to Case~3.
\end{enumerate}\vspace{-.6cm}
\end{proof}

\section{Analysis of the kernel size} \label{sec:analysis}

We will show that a graph reduced under our four rules has size linear in $|D|$, the size of a solution. To this aim, we assume that the graph is \emph{plane} (that is, given with a fixed embedding).
We recall that an \emph{embedding} of a graph $G=(V,E)$ in the plane $\mathbb{R}^2 $ is a function $\pi : V \cup E \mapsto \mathcal{P}(\mathbb{R}^2) $, which maps each vertex to a point of the plane and each edge to a simple curve of the plane, in such a way that the vertex images are pairwise disjoint, and each edge image corresponding to an edge $\{u,v\}$ has as endpoints the vertex images of $u$ and $v$, and does not contain any other vertex image. An embedding is \emph{planar} if any two edge images may intersect only at their endpoints. In the following, for simplicity, we may identify vertices and edges with their images in the plane. Following Alber~\emph{et al}.~\cite{AFN04}, we will define a notion of region in an embedded graph adapted to our definition of neighborhood, and we will show that, given a solution $D$, there is a maximal region decomposition $\Re$ such that:
\begin{itemize}\itemsep0em
  \item $\Re$ has $O(|D|)$ regions,
  \item $\Re$ covers all vertices, and
  \item each region of $\Re$ contains $O(1)$ vertices.
\end{itemize}
The three following propositions treat respectively each of the above claims.


We now define our notion of region, which slightly differs from the one defined in~\cite{AFN04}.

\begin{defi}\label{def:region}
Let $G = (V_B \cup V_R,E) $ be a plane graph and let $v,w \in V_B$ be a pair of distinct blue vertices.
A \emph{region} $R(v,w)$ between $v$ and $w$ is a closed subset of the plane such that:
\begin{itemize}\itemsep0em
  \item the boundary of $R(v,w)$ is formed by two simple paths connecting $v$ and $w$, each of them having at most 4 edges,
  \item all vertices (strictly) inside $R(v,w)$ belong to $N(v,w)$ or $N(N(v,w))$, and
  \item the complement of $R(v,w)$ in the plane is connected.
\end{itemize}

We denote by $\partial R(v,w)$ the boundary of $R(v,w)$ and by $V(R(v,w))$ the set of vertices in the region (that is, vertices strictly inside, on the boundary, and the two extremities $v,w$). The \emph{size} of a region is $|V(R(v,w))|$.

Given two regions $R_1(v_1,w_1)$ and $R_2(v_2,w_2)$, we denote by $R_1(v_1,w_1) \cup R_2(v_2,w_2)$  the union of the two closed sets in the plane, and by $R_1(v,w) \uplus R_2(v,w)$ the special case where the union defines a region; note that this latter case can occur only if the two regions share both extremities and one path of their boundaries. Note also that the boundary of $R_1(v,w) \uplus R_2(v,w)$ is defined by the two other paths of the boundaries.
\end{defi}

Note that a subgraph defining a region has diameter at most 4, to be compared with diameter at most 3 in~\cite{AFN04}. We would like to point out that the assumption that vertices $v$ and $w$ are distinct is not necessary, but it makes the proofs easier. In the following, whenever we speak about a pair of vertices we assume them to be  distinct. Note also that we do not assume that the two paths of the boundary of a region are edge-disjoint or distinct, hence in particular a path corresponds to a degenerated region.

We want to decompose a reduced graph into a set of regions which do not overlap each other. To formalize this, we provide the following definition of \emph{crossing} regions, which can be seen as a more formal and precise version of the corresponding definition given in~\cite{AFN04}.

Recall that we are considering a plane graph, hence for each vertex $v$, the embedding induces a circular ordering on the edges incident to $v$. We first need the (recursive) definition of \emph{confluent} paths.

\begin{defi}\label{def:confluent} Two simple paths $p_1,p_2$ in a plane graph $G$ are \emph{confluent} if:
\begin{itemize}\itemsep0em
\item they are vertex-disjoint, or
\item they are edge-disjoint and for each vertex $v \in p_1 \cap p_2$ distinct from the  extremities, among the four edges of $p_1,p_2$ containing $v$, the two edges in $p_1$ are consecutive in the circular ordering given by the embedding (hence, the two edges in $p_2$ are consecutive as well), or
\item the two paths obtained by contracting common edges are confluent.
\end{itemize}\end{defi}

Note that, by definition, a path is confluent with itself. In Definition~\ref{def:confluent}, whenever an edge is contracted, the planar embedding of $G$ is modified in the natural way (if multiple edges or loops appear, they can be safely removed).

\begin{defi}\label{def:cross} Two distinct regions $R_1, R_2$ in the plane \emph{do not cross} if:
\begin{itemize}\itemsep0em
\item  $(R_1 \setminus \partial R_1) \cap R_2 = (R_2 \setminus \partial R_2) \cap R_1 = \emptyset $ (i.e., the interiors of the regions are disjoint), and
\item any path $p_1$ in $\partial R_1$ is confluent with any path $p_2$ in $\partial R_2$.
\end{itemize}
\noindent Otherwise, we say that $R_1, R_2$ \emph{cross}. If two regions cross because of two paths $p_1 \in \partial R_1$ and $p_2 \in \partial R_2$ that are not confluent, we say that these regions \emph{cross on $v \in  p_1 \cap p_2$} if:
\begin{itemize}\itemsep0em
\item $v$ does not satisfy the ordering condition of Definition~\ref{def:confluent}, or
\item $v$ is an extremity of an edge $e$ such that in $G/e$ (i.e., the graph obtained from $G$ by contracting $e$), $R_1$ and $R_2$ cross on the vertex resulting from the contraction.
\end{itemize} \end{defi}

Of course, two regions can cross on many vertices. We use the latter condition of Definition~\ref{def:cross} in the case of degenerated regions (that is, paths), where only this condition may hold.

We now have all the material to define what a region decomposition is. Compared to Alber \emph{et al.}~\cite[Definition 3]{AFN04}, we have two additional conditions to be satisfied by a \emph{maximal} decomposition, which are in fact conditions satisfied by the region decomposition constructed by the greedy algorithm presented in~\cite{AFN04}.

\begin{defi}\label{def:regionDec} Let $G = (V_B \cup V_R,E) $ be a plane graph and let $D \subseteq V_B$.
A \emph{$D$-decomposition} of $G$ is a set of regions $\Re$ between pairs of vertices in $D$ such that:
\begin{itemize}\itemsep0em
  \item any region $R \in \Re$ between two vertices $v$, $w$ is such that $V(R)\cap D =\{v,w\}$, and
  \item any two regions in $\Re$ do not cross.
\end{itemize}

We denote $V(\Re) = \bigcup _{R\in\Re} V(R)$. A $D$-decomposition is \emph{maximal} if there are no regions
\begin{itemize}
\item $R \notin \Re$ such that $\Re \cup \{R\}$ is a $D$-decomposition with $V(\Re) \subsetneq V(\Re \cup \{R\})$,
\item $R \in \Re$  and $R' \notin \Re$ with $V(R) \subsetneq V(R')$ such that $\Re \cup \{R\} \setminus \{R'\}$ is a $D$-decomposition, or
\item  $R_1,R_2 \in \Re$ such that $\Re \cup \{R_1 \uplus R_2\} \setminus \{R_1,R_2\}$ is a $D$-decomposition.
\end{itemize}

\end{defi}


In order to bound the number of regions in a decomposition, we need the following definition. We consider multigraphs without loops.

\begin{defi}\label{def:thin}
A planar multigraph $G$ is \emph{thin} if there is a planar embedding of $G$ such that for any two edges $e_1,e_2$ with identical endvertices,
there is a vertex image inside the two areas enclosed by the edge images of $e_1$ and $e_2$. In other words, no two edges are homotopic.
\end{defi}

In~\cite[Lemma 5]{AFN04}, the bound on the number of regions in a decomposition relies on applying Euler's formula to a thin graph. Since it appears that some arguments are missing in that proof, for completeness we provide here an alternative proof.

\begin{lem}\label{fait_thin}
If $G=(V,E)$ is a thin planar multigraph with $|V| \geq 3$, then $|E| \leq 3|V| -6$.
\end{lem}

\begin{proof}
Recall that a triangulated simple graph is a maximal planar graph, that is, a graph where all faces contain exactly 3 edges.
For a triangulated (connected) simple graph $H=(V_H,E_H)$, Euler's formula~\cite{Die05} states that  $|E_H| =  3|V_H| - 6$. We proceed to extend the notion of triangulation to thin multigraphs and we will show that Euler's formula still holds.
Let the \emph{size} of a face be the number of edges it contains. Given a thin planar multigraph $G=(V,E)$, we define recursively a triangulation $G'=(V,E')$ of $G$ as follows. For each face $f$ of size 4 or more, we add arbitrarily an edge between two non-adjacent vertices. Note that two such vertices always exist. Indeed, otherwise $f$ would contain four vertices $v_1,v_2,v_3,v_4$ in this cyclic order around $f$, such that both edges $\{v_1,v_3\}$ and $\{v_2,v_4\}$ are drawn in the same region of the complement of $f$; this would contradict the planarity of the embedding. As $G$ is thin, note there is at least one vertex inside each face of size 2. Then we add the two edges between each such inner vertex and the two vertices defining the face. We say that a multigraph is \emph{triangulated} whenever no more edges can be added.
Now, given a triangulated planar multigraph $G'=(V,E')$, we transform it into a triangulated planar simple graph $H=(V_H,E_H)$ as follows. As far as there exists a multiple edge between two vertices $u$ and $v$, let $e$ be an occurrence  of this edge. We know that $e$ belongs to two faces of size 3 containing vertices $x,u,v$ and $y,v,u$ respectively, for two vertices $x$ and $y$. Then we subdivide $e$ into two edges $\{u,w_e\}, \{v,w_e\}$, where $w_e$ is a new vertex, and we add the edges $\{x,w_e\}$ and  $\{y,w_e\}$. Let $p$ be number of vertices added during this procedure. Note that $|V_H| = |V| + p$ and $|E_H| = |E'| + 3p$. Since Euler's formula holds for $H$, we have that $|E'| + 3p = 3(|V| + p) - 6$, or equivalently, $|E'| = 3|V|  - 6$. Therefore, it holds that $|E| \leq |E'| = 3|V|  - 6$, and the lemma follows.
\end{proof}

We would like to point out that it is possible to prove that any vertex in a reduced graph is on a path with at most 4 edges connecting two dominating vertices. Since in what follows we will use several restricted variants of this property, we will provide an ad-hoc proof for each case.


\begin{prop} \label{prop_nb_reg}
Let $G$ be a reduced plane graph and let $D$ be a \drb in $G$ with $|D| \geq 3$.
There is a maximal $D$-decomposition of $G$ such that $|\Re| \leq 3 \cdot |D| - 6$.
\end{prop}

\begin{proof} The proof strongly follows the one of Alber \emph{et al.}~\cite[Lemma 5 and Proposition 1]{AFN04}. Even if our definition of region is different, we shall show that the same algorithm that they present can be used to construct such a $D$-decomposition. Nevertheless, we will provide some arguments concerning planarity that were missing in~\cite{AFN04}.

We consider the algorithm that, for each vertex $u$, adds greedily to the decomposition $\Re$ a region $R$ between any two vertices $v,w \in D$ (possibly, $v$ or $w$ may be equal to $u$), containing $u$,
not containing any vertex of $D \setminus \{v,w\}$,
not crossing any region of $\Re$, and of maximal size, if it exists.
By definition, $\Re \cup \{R\}$ is a region decomposition, and by greediness and because regions are chosen of maximal size, the decomposition is maximal according to Definition~\ref{def:regionDec}.

In order to apply Lemma~\ref{fait_thin}, we proceed to define a multigraph that we will prove to be thin. Let $G_\Re =(D, E_\Re)$ be the multigraph with vertex set $D$ and with an edge $\{v,w\}$ for each region in $\Re$ between two dominating vertices $v$ and $w$. Let $\pi$ be the embedding of the plane graph $G$, and we consider the embedding $\pi_\Re$ of $G_\Re$ such that for $v \in D$, $\pi_\Re(v) = \pi(v)$,  and for $e \in E_\Re$ corresponding to a region $R \in \Re$ with $p$ an arbitrary boundary path of $R$, $\pi_\Re(e) = \bigcup_{f \in p} \pi(f)$. \modifOK{(note that such a path does not contain inner dominating vertices, hence $\pi_\Re(e)$ does not contain vertex images). For an edge set $F \subseteq E_\Re$, we denote $\pi_\Re(F) = \bigcap_{e\in F} \pi_\Re(e) $. If the constructed embedding is not planar, we proceed to modify it in order to make it planar as follows. Given $F,F' \subseteq E_\Re$, observe that if $F \subseteq F'$ then $\pi_\Re(F) \supseteq \pi_\Re(F')$.}{We modify the constructed embedding $\pi_\Re$ in order to make it planar. First, observe that images $\pi_\Re(v),\pi_\Re(w)$ are distinct for two distinct vertices $v,w \in D$. Secondly, observe that $\pi_\Re(e)$ for $e \in E_\Re$ does not contain the image of a vertex, because  by definition the corresponding path in $G$ does not contain a dominating vertex. Hence, if $\pi_\Re$ is not planar this is due to an edge intersection. If such an intersection exists, we proceed as follows. For an edge set $F \subseteq E_\Re$, we denote $\pi_\Re(F) = \bigcap_{e\in F} \pi_\Re(e) $.} As far as there
exists an edge set $F\subseteq E_\Re$ such that $\pi_\Re(F) \nsubseteq \bigcup_{v \in D} \pi_\Re(v)$, we apply the following procedure in the inclusion order of such edge sets, starting with a maximal such set of edges $F$.
Consider the subset of the plane $C(F)$ containing, for each $x\in \pi_\Re(F) \setminus \bigcup_{v \in D} \pi_\Re(v)$, a closed ball of center $x$ and radius $\epsilon_x$ for a sufficiently small real number $\epsilon_x > 0$,
such that $C(F)$ intersects only curves corresponding to edges in $F$. \modifOK{; note that such a subset of the plane exists, since $C(F)$ forms a sharp corner around dominating vertices \ig{il faudrait donner un petit argument de pourquoi $C(F)$ fait un ``sharp corner''} and does not contain them, and since we proceed by inclusion order\\ \ig{pourquoi le faire par ``inclusion order'' garantit l'existence d'$\epsilon$?}.}{Such a subset $C(F)$ of the plane exists. Indeed, since we proceed by inclusion order, the considered curve $\pi_\Re(F)$ does not intersect $\pi_\Re(e)$ for any $e \in E_\Re \setminus F$, except for vertices of $D$; hence for all $x \in \pi_\Re(F) \setminus \bigcup_{v \in D} \pi_\Re(v) $ there exists an $\epsilon_x$ such that the ball of center $x$ and radius $\epsilon_x$ does not intersect any $\pi_\Re(e)$.} Since the paths defining the borders of regions in $\Re$ are confluent (that is,  in the intersection of several paths on a vertex $v$, for each subset of two of these paths, the edges of each path are consecutive around $v$), it is easy to see that for each edge $e \in F$ there is a connected curve $C_e$ inside $C(F)$ disjoint from all the edge images of $\pi_\Re$ except for the two endpoints of $\pi_\Re(e) \cap C(F)$. Hence we can replace, for each edge $e$ in such a set $F$, the edge image $\pi_\Re(e)$ in $C(F)$ with the curve $C_e$. When there does not exist such an edge set $F$ anymore, the obtained embedding of $G_\Re$ is planar, since in that case any two edge images may intersect only at their endpoints. This re-embedding procedure is schematically illustrated in Figure~\ref{fig_replonge}.

\begin{figure}[h]
\begin{center}
   \includegraphics[scale=1]{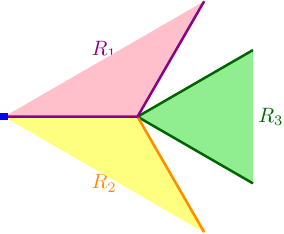}
   \includegraphics[scale=1]{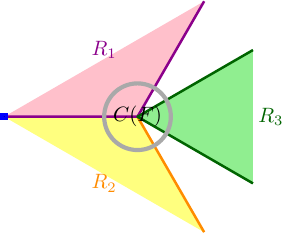}
   \includegraphics[scale=1]{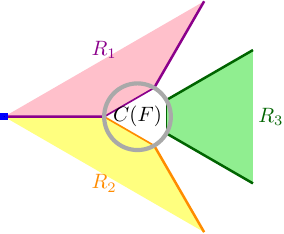}
   \includegraphics[scale=1]{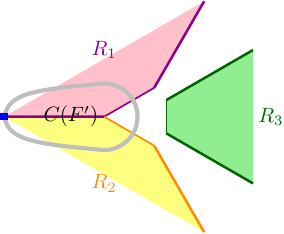}
   \includegraphics[scale=1]{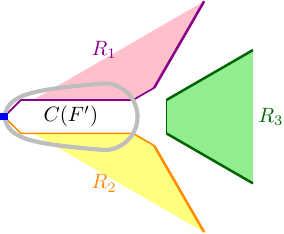}
   \includegraphics[scale=1]{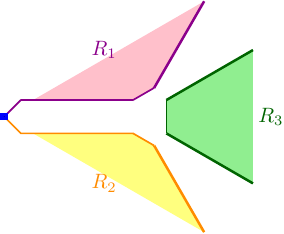}

\end{center}
   \caption{An example of the re-embedding procedure. We consider three regions $R_1,R_2,R_3$ (filled with light colors), where the three considered edges of $G_\Re$ (depicted with dark colors) are initially embedded on the boundary of the regions. Let $F$ be the set of these three edges. First, we modify the intersection  of the edges in $F$, which is a point. Let $F'$ be the set of two edges corresponding to $R_1$ and $R_2$. Then, we modify the intersection of the edges in $F'$, which is a segment. At the end of this procedure, the edges are pairwise disjoint except possibly for their extremities.}
   \label{fig_replonge}
\end{figure}

We will now prove that the multigraph $G_\Re$ is thin, that is, for each pair of edges $e_1,e_2 \in E_\Re$ between the same pair of vertices $v,w$ (corresponding to regions $R_1,R_2$ and embedded following  paths $p_1,p_2$, respectively) there is a vertex of $D$ in both open subsets of the plane enclosed by $e_1,e_2$. This will allow us to apply \rlem{fait_thin}, implying that the constructed decomposition has at most $3 |D| - 6$ regions. Let $O_\Re$ be one of these two open sets.
Let $O$ be the open set enclosed by $p_1,p_2$ and corresponding to $O_\Re$ (in the sense that the same orientation is chosen to traverse $e_1,e_2$ and $p_1,p_2$) minus the two regions $R_1,R_2$.
Note that the symmetric difference of $O_\Re$ and $O$ consists of \modifOK{some of the two regions \ig{tu veux dire toute une r\'egion, ou juste une partie d'une des r\'egions?}}{an element of $\{ \emptyset, R_1, R_2, R_1\cup R_2 \}$ (which may be in $O_\Re$ but not in $O$)}  and some of the subsets $C(F)$ used to re-embed $e_1,e_2$ \modifOK{}{(which may be in $O$ but not in $O_\Re$)}. Since $C(F)$ does not contain a dominating vertex (for any \modifOK{$F \in \mathcal{F}_{e_1,e_2}$}{$F$ containing $e_1$ or $e_2$}), every dominating vertex in $O$ is in \modifOK{$O'$ \ig{on n'a pas d\'efinit $O'$}}{$O_\Re$}. \modifOK{And since $O'$ is a subset of $O_\Re$, every dominating vertex in $O$ is in $O_\Re$.}{}Therefore, in order to prove that there is a vertex of $D$ in the set $O_\Re$, it suffices to prove that there is a vertex of $D$ in the set $O$.
These two definitions are illustrated in Figure \ref{fig_ouvert}.

\begin{figure}[h]
\begin{center}
   \includegraphics[scale=1]{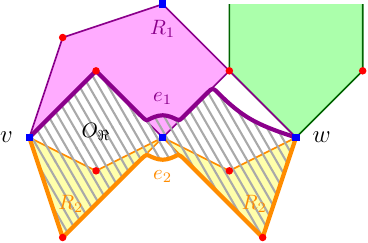} (a)
   \includegraphics[scale=1]{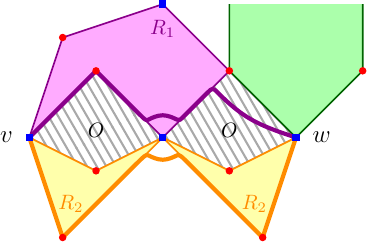} (b)
\end{center}
   \caption{An illustration of the definition of $O_\Re$\modif{$, O'$,}{} and $O$ in the proof of Proposition~\ref{prop_nb_reg}. The \modifOK{}{figures show the} two regions $R_1$ and $R_2$ of $G$ (filled with light colors) and the two corresponding edges $e_1$ and $e_2$ of $G_\Re$ \modif{are depicted in (a)}{}. Note that the re-embedding procedure has been applied twice: firstly on the blue vertex common to $R_1$ and $R_2$ and secondly on an edge common to $R_1$ and a third region. The open set $O_\Re$ enclosed by $e_1$ and $e_2$ (the \modifOK{dark}{dashed} area) is depicted in \modifOK{(b)}{(a)}.\modif{The open set $O'$ (the dark area) is depicted in (c). Finally,}{} The open set $O \subseteq \overline{R_1 \cup R_2}$ (the \modifOK{dark}{dashed} area) is depicted in \modifOK{(d)}{(b)}. }
   \label{fig_ouvert}
\end{figure}

 Note that $O$ is not empty, since otherwise $R_1$ and $R_2$ would share an entire path of their boundaries, which contradicts the maximality of the decomposition, as in that case $R_1, R_2$ could be replaced with $R_1 \uplus R_2$.

 Let us assume for the sake of contradiction that there is no vertex of $D$ in $O$. We distinguish three cases:
\begin{itemize}\itemsep0em

\item If $O$ intersects a region $R_3 \in \Re$, then $R_3$ is necessarily between $v$ and $w$, as otherwise $R_3$ would cross $R_1$ or $R_2$. In this case, we can recursively apply the same argument to $R_1,R_3$ and $R_2,R_3$. If $u_i \in V(R_i) \cap N(v)$ for $i \in\{1,2,3\}$, according to the circular order around $v$ we have that $u_1<u_3<u_2$  (similarly around $w$). Since the degrees of $v$ and $w$ are finite, so is the number of considered regions $R_3$ in this recursive argument, which therefore terminates.

\item Otherwise, assume first that $O$ does not contain any blue vertex.  Then the red vertices in $O$ (if any) must be dominated by $v$ or $w$. Hence, since we are assuming that  $O$ does not intersect any region  of $\Re$, it follows that $R_1 \cup O \cup \partial(O)$  is a larger region enclosed by a path of $R_1$ and a path of $R_2$, where $\partial(O)$ denotes the boundary of the open set $O$. We have a contradiction with the maximality of $\Re$.

\item Otherwise, if $O$ contains at least one blue vertex $b \notin D$, we shall show that $b$ lies on a path $p =\{v,r,b,r',w\}$ for some vertices $r,r'$. Indeed, since $G$ is reduced under \rrgl{rgl_bleu}, $N(b) \neq \emptyset$, there is some vertex $r \in N(b)$ that is dominated, without loss of generality, by $v$ and not by $w$. Again, by \rrgl{rgl_bleu}, $b$ and $v$ have incomparable neighborhoods, so there is some $r' \in N(b) \setminus N(v)$ that is dominated by $w$ and not by $v$. Notice that $r,r'$ are in $O$ or in its boundary, hence $P = \{v,r,b,r',w\}$ is a region which does not cross $R_1,R_2$, a contradiction with the maximality of $\Re$.

\end{itemize}\vspace{-.7cm}
\end{proof}

\begin{prop} \label{prop_nb_excl}
Let $G=(V_B \cup V_R,E)$ be a reduced plane graph and let $D$ be a \drb in $G$ with $|D| \geq 3$.
If $\Re$ is a maximal $D$-decomposition, then $V_B \cup V_R\subseteq V(\Re )$. 
\end{prop}

\begin{proof} The proof again follows that of Alber \emph{et al.}~\cite[Lemma 6 and Proposition 2]{AFN04}, where similar arguments are used to bound the number of vertices which are not included in a maximal region decomposition. We have to show that all vertices are included in a region of $\Re$, that is, $V_R \cup V_B \subseteq V(\Re)$.

Since $N(D)$ covers $V_R$, it holds that $V_R = \bigcup_{v \in D} N(v)$. We proceed to prove that $N(v) \subseteq V(\Re)$ for all $v \in D$. Let $v\in D$ and let $ r \in N(v)$. We now show that there is a path $p=\{v,r,\dots,w\}$ with $w \in D$ and with at most four edges. Indeed, since \rrgl{rgl_som} cannot be applied on $v$, then $r \notin P(v)$, and by definition of private neighborhood there are two vertices $b \in N(r)$ and $r' \in  N(b) \setminus N(v)$. If $b \in D$, then $P = \{v,r,b\}$, with $b = w$,  is the desired path. Otherwise, $r'$ is dominated by some vertex $w\in D$, and then $P = \{v,r,b,r',w\}$.

Assume for contradiction that $r \notin V(\Re)$. It follows that $P \nsubseteq R $ and $P$ does not cross $R$ on $r$ for any $R \in \Re$. We distinguish two cases depending on the length of $P$:

\begin{itemize}\itemsep0em
\item If $P=\{v,r,w\}$ for some vertex $w$, then $P$ can be added to $\Re$, which contradicts the maximality of $\Re$.
\item If $P=\{v,r,b,r',w\}$ for some vertices $b,r',w$ with $w \neq v$, then either $P$ can be added, which again contradicts the maximality of $\Re$, or $P$ crosses some region $R(x,y)$ of $\Re$. Recall that we assume $r \notin V(R(x,y))$. We distinguish two cases (see Figure~\ref{fig_prop} for an illustration):
\begin{itemize}\itemsep0em

\item[$\circ$] If $P$ and $R(x,y)$ cross on $b$, then $b$ is on $\partial R(x,y)$. Let $r''$ be a vertex on $\partial R(x,y)$ such that the edge $\{b,r''\}$ is the successor of the edge $\{b,r\}$ in the circular order defined by the embedding. In this case, we consider the path $p' =\{v,r,b,r'',x\}$, where we have assumed without loss of generality that $r''$ is a neighbor  of $x$; the case where $r''$ is a neighbor  of $y$ is symmetric.
\item[$\circ$] Otherwise, necessarily $P$ crosses a region $R(x,y) \in \Re$ on $r'$, and then $r'$ is on $\partial R(x,y)$. Assume without loss of generality that $r' \in N(x)$. In this case, we consider the path $p'=\{v,r,b,r',x\}$.
\end{itemize}
In both cases, either $P'$ can be added to $\Re$, which contradicts the maximality of $\Re$, or $P'$ crosses another region and we can apply recursively the same argument. Again, the recursion must be finite, as $r<r''<r'$ and $b<x<w$ in the circular order around $b$ and $r'$, respectively, and the degrees of $b$ and $r'$ are finite.
\end{itemize}
So $\bigcup_{v \in D} N(v) \subseteq V(\Re)$, as we wanted to prove.

\begin{figure}[h]
\begin{center}
   \includegraphics[scale=0.9]{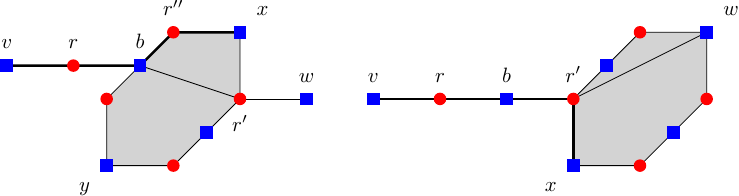}
\end{center}
   \caption{Illustration of the two ways that the path $\{v,r,b,r',w\}$, as defined in Proposition~\ref{prop_nb_excl}, can cross a region. Blue (resp. red) vertices are depicted with $\textcolor{blue}{\blacksquare}$ (resp. \LARGE{$\textcolor{red} {\bullet}$}\normalsize{)}. }
   \label{fig_prop}
\end{figure}

We finally show that $V_B \subseteq V(\Re)$. Recall that we assume that $|D| > 2$.  We consider separately vertices in $V_B \setminus D$ and vertices in $D$.

Let first $b \in V_B \setminus D$. Since $G$ is reduced,  $b$ is neighbor of two red vertices $r'$ and $r''$ dominated respectively by $v$ and $w$ with $v \neq w$, as otherwise vertex $b$ could be removed by \rrgl{rgl_bleu}. We consider the (degenerated) region $\{v,r',b,r'',w\}$, and with an argument similar to the one given above, if we assume that $b \notin V(\Re)$ we obtain a contradiction.
Let then $v \in D$. By \rrgl{rgl_som}, $v$ cannot be a single dominating vertex in a connected component. Hence there is a vertex $w \in D$ at distance at most 4 from $v$. We consider a path between $v$ and $w$ as a region, and once again we obtain a contradiction using similar arguments. So $V_B \subseteq V(\Re)) $.

Therefore, all the vertices of $G$ belong to the decomposition $\Re$, as we wanted to prove.
\end{proof}


\begin{prop}\label{prop_nb_incl}
Let $G=(V_B \cup V_R,E)$ be a reduced plane graph, let $D$ be a \drb in $G$, and let $v,w \in D$.
Any region $R$ between $v$ and $w$ contains at most $ 14 $ vertices distinct from $v$ and $w$.
\end{prop}

\begin{proof}
Let $R$ be an arbitrary region between $v$ and $w$. Since $G$ is reduced under \rrgl{rgl_pair}, $|P(v,w)| \leq 2$. And, since $G$ is reduced under Rules~\ref{rgl_bleu} and~\ref{rgl_rouge}, each vertex strictly inside $R$ has a neighborhood incomparable with the neighborhood of any other vertex. It will become clear from the proof  that the worst bound is given by the case when $\partial R$ is as large as possible, that is, when it contains 8 vertices, which will be henceforth denoted by $v,r_v,b,r_w,w,r_w',b'$, and $r_v'$.

Let us first bound the number of non-private red neighbors of $v$ and $w$ in $R$.

\begin{claimN}\label{fact:non-private}
There are at most 4 vertices from  $N(v,w) \setminus P(v,w)$ strictly inside $R$.
\end{claimN}
\begin{proof}
Let $s$ be a non-private red vertex. The neighborhood of $s$ contains $v$ or $w$ (because $s \in N(v,w)$), $b$ or $b'$ (because $ s \notin P(v,w)$), and  at least another blue vertex (because $N(s)$ has to be incomparable with $N(r_v), N(r_w), N(r'_v)$, and $N(r'_w)$).

Assume for contradiction that there are two non-private red vertices $s$ and $s'$ strictly inside $R$ such that $\{v,b\} \subseteq N(s) \cap N(s')$. (By symmetry, the same argument applies to $\{v,b'\}, \{w,b\}$, or $\{w,b'\}$ instead of $\{v,b\}$.) Since both $s$ and $s'$ are neighbors of $v$ and $b$, by planarity one of them, say $s'$ cannot be adjacent to $w$ nor $b'$. Therefore, since $N(s')$ has to be incomparable with $N(s)$, there should exist a vertex $t' \in N(s') \setminus N(s)$, which again by planarity cannot be neighbor of any other red vertex in $R$, except possibly $r_v$. But then $N(t') \subseteq N(v)$, and therefore vertex $t'$ should have been deleted by Rule~\ref{rgl_bleu}, a contradiction. Thus, vertex $s'$ cannot exist.

Summarizing the above discussion, it holds that any red vertex in $N(v,w) \setminus P(v,w)$ has to be neighbor of $v$ or $w$, and of $b$ or $b'$, and any two such red vertices  cannot have simultaneously a common neighbor in the set $\{v,w\}$ and in the set $\{b,b'\}$. Hence, there can be at most 4 red vertices in $N(v,w) \setminus P(v,w)$ distinct from $r_v,r'_v,r_w,r'_w$, with neighbors $\{v,b\}, \{v,b'\}, \{w,b\}$, and $\{w,b'\}$, respectively. This configuration is depicted in Figure~\ref{fig_region}.
\end{proof}

\begin{figure}[t]
\vspace{-1.0cm}
\begin{center}
   \includegraphics{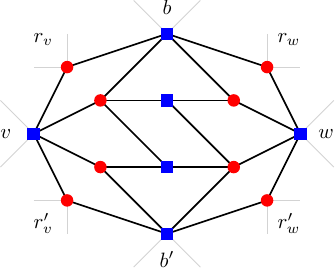}(a)
   \includegraphics{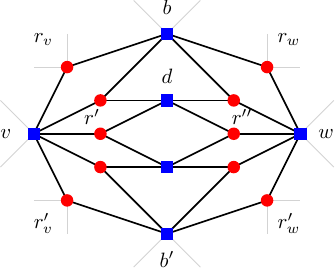}(b)
   \includegraphics{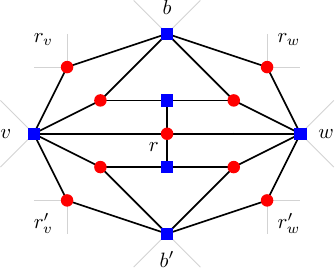}(c)
   \includegraphics{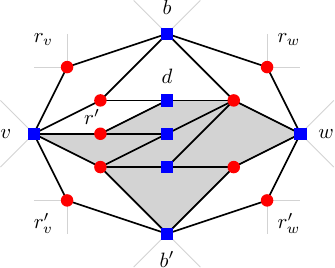}(d)
\end{center}
\vspace{-.4cm}
   \caption{Examples of a worst cases in the proof of Proposition~\ref{prop_nb_incl}.
   Blue (resp. red) vertices are depicted with $\textcolor{blue}{\blacksquare}$ (resp. \LARGE{$\textcolor{red} {\bullet}$}\normalsize{)}.
   In (a) the \rrgl{rgl_pair} is not applied.
   In (b) the \rrgl{rgl_pair} Case~1 is applied.
   In (c) the \rrgl{rgl_pair} Case~2 is applied.
   In (c) the \rrgl{rgl_pair} Case~3 is applied.
   The global worst cases correspond to (b) and (d).}
   \label{fig_region}
\end{figure}

It just remains to bound the number of blue vertices strictly inside $R$, and to this end we  distinguish five cases which correspond to the case where \rrgl{rgl_pair} is not applied, plus the four cases of this rule.

From the proof that follows, it will be easy to see that  the maximum number of blue vertices in $R$ is achieved when the number of non-private red vertices in the interior of $R$ is also maximum; this number is 4 by Claim~\ref{fact:non-private}. So we assume henceforth that $R$ contains $4$ non-private red vertices, and from the proof of Claim~\ref{fact:non-private} it follows that their neighborhoods in the boundary of $R$ are as depicted in Figure~\ref{fig_region}.   These 4 red vertices together with their incident edges toward the boundary split the region $R$ into $5$ subregions (we use the term \emph{subregion} for convenience, but it has nothing to do with the definition of region). Note that only one of these subregions, say $R'$, contains both $v$ and $w$. Since the graph is reduced under Rule~\ref{rgl_bleu} (similarly to the proof of Claim~\ref{fact:non-private}), it follows that only the subregion $R'$ can contain blue vertices. Thus, it just remains to bound the number of of blue vertices that can be contained in $R'$.
\begin{enumerate}
\item [0.]
Assume first that \rrgl{rgl_pair} has not been applied on  $v,w$ (that is, $P(v,w) = \emptyset$).
Taking into account that the neighborhoods of these vertices have to be incomparable, Figure~\ref{fig_region-NEW} shows exhaustively the possible configurations that respect planarity, where the darker area corresponds to the subregion $R'$; there can be at most $2$ blue vertices strictly inside $R$.

    \begin{figure}[h!tb]
\begin{center}
   \includegraphics[scale=.9]{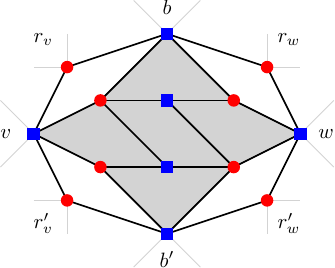}
   \includegraphics[scale=.9]{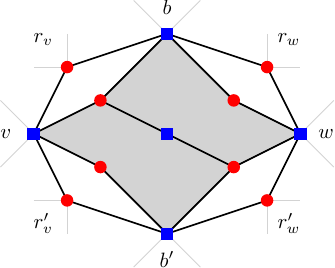}
   \includegraphics[scale=.9]{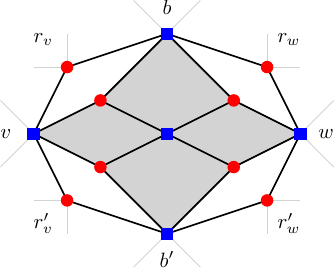}
\end{center}
\vspace{-.4cm}
   \caption{Possible configurations when $|P(v,w)| = 0$ in the proof of Proposition~\ref{prop_nb_incl}.}
   \label{fig_region-NEW}
\end{figure}

\item
Assume now that Case~1 of \rrgl{rgl_pair} has been applied on $v,w$, and let $r',r''$ be the two private neighbors of $v,w$ (that is, $P(v,w) =  \{r',r''\}$).\\

 \begin{claimN}\label{claim:case1nD}
 There is no blue vertex from $V_B \setminus \mathcal{D}$ strictly inside $R'$.
    \end{claimN}

    \begin{proof}
    Note that $\mathcal{D} \neq \emptyset$, as otherwise $r',r''$ have degree 1, and \rrgl{rgl_som} has removed $v$ and $w$. Let $d \in\mathcal{D}$. The path $\{v,r',d,r'',w \}$ splits $R$ into two areas (see Figure~\ref{fig_region}(b)): in the first one all non-private red vertices are adjacent to $b$,  while in the second one all non-private red vertices are adjacent to $b'$.

    Assume for contradiction that there is a blue vertex $b \in V(R') \setminus (\mathcal{D}\cup\{v,w,b,b'\})$ in one of areas described above.
    Since $b$ is not adjacent to $r',r''$, and by planarity, its neighborhood would be included in $N(b)$ or in $N(b')$, contradicting the incomparability of neighborhoods.
    \end{proof}

    \begin{claimN}\label{claim:case1D}
    There are at most 2 blue vertices in $\mathcal{D}$.
    \end{claimN}

    \begin{proof}
     Assume for contradiction that $|\mathcal{D}| \geq 3$, and let $\{d_1,d_2,d_3\} \subseteq \mathcal{D}$. Observe that $d_i$ is adjacent to $r'$, $r''$, and at least another red vertex (because $N(d_i) \neq N(D_j)$ for $i \neq j$), for $i \in \{1,2,3\}$.

     Note that the graph $G[V(R) \setminus (\{r',r''\} \cup \mathcal{D})]$ is connected, as all its vertices distinct from $v$ and $w$ are neighbors of at least one of them, and $v$ and $w$ are linked by a path of the boundary of $R$. Let $M$ be the graph obtained from $G[V(R)]$ by contracting the (connected) subgraph $G[V(R) \setminus (\{r',r''\} \cup \mathcal{D})]$ into a single vertex, say $c$. Note that the vertex set of $M$ consists of $r',r'',c$, and the vertices in $\mathcal{D}$, and that by construction $M$ is a minor of $G$. Recall that for $i \in \{1,2,3\}$, $N(d_i) \cap (V(R) \setminus \{r',r''\}) \neq \emptyset$, which is equivalent to saying that for $i \in \{1,2,3\}$ vertex $d_i$ is adjacent to vertex $c$ in $M$. It follows that $M$ contains a subgraph isomorphic to $K_{3,3}$ defined by the bipartition $\{r',r'',c\}$ and $\{d_1,d_2,d_3\}$, which is also a minor of $G$, contradicting by Kuratowski's Theorem~\cite{Die05} the hypothesis that $G$ is a planar graph.
\end{proof}

\item
Assume now that Case~2 of \rrgl{rgl_pair}  has been applied on $v,w$, and let $r$ be the private neighbor of $v,w$ (that is, $P(v,w) =  \{r\}$). 
Note that the path $\{v,r,w \}$ splits $R$ into two areas; see Figure~\ref{fig_region}(c). Similarly to the argument of Claim~\ref{claim:case1nD} in the case above, it easily follows that there is no other blue vertex from $V_B \setminus \mathcal{D}$ inside any of these 2 areas.

    \begin{claimN}\label{claim:case2D}
    There are at most 2 blue vertices in $\mathcal{D}$.
    \end{claimN}

    \begin{proof} As each vertex $d \in \mathcal{D}$ has incomparable neighborhood with $N(v)$ and $N(w)$, necessarily $N(d)$ contains $r$, a red vertex in $N(v,w) \setminus N(w)$, and another in $N(v,w) \setminus N(v)$. Note that in each of the two areas described above, there is an unique vertex of each type, hence there is a unique vertex from $\mathcal{D}$, in each of the two areas.
    \end{proof}

\item
Assume now that Case~3 of \rrgl{rgl_pair}  has been applied on $v,w$, and let $r'$ be the private neighbor of $v,w$ (that is, $P(v,w) =  \{r'\}$).
Note that $\mathcal{D}$ contains at least one vertex, since $G$ is reduced under \rrgl{rgl_som}. Let $d \in \mathcal{D}$.
Necessarily, $N(d)$ contains $r'$ and another red vertex in $N(v,w) \setminus N(v)$; let $\bar r_w$ be this vertex. The path $\{v,r',d,\bar r_w\}$ splits $R'$ into two areas: one containing $b$ and one containing $b'$. Without lost of generality, we can assume that $\bar r_w$ is adjacent to $b$. According to the arguments in the proof of Claim~\ref{claim:case1nD}, there is no blue vertex from $V_B\setminus \mathcal{D}$ in the area containing $b$, and we can choose $d$ such that this area contains no vertex from $\mathcal{D}$. Hence, it remains to bound the number of vertices in the area containing $b'$. Taking into account that the neighborhoods of these vertices have to be incomparable, the possible configurations that respect planarity can be enumerated exhaustively. These configurations are the same as the ones of Case~0 and Figure~\ref{fig_region-NEW}. It follows that there can be at most two vertices strictly inside this area.

\item
Symmetrically to Case~3.

\end{enumerate}

It follows that a region contains at most $14$ vertices distinct from $v,w$.

\vspace{-.5cm}
\end{proof}

\vspace{.65cm}

We are finally ready to piece everything together and prove Theorem~\ref{th:main}.

\vspace{.35cm}


\begin{proofof}
Let the input consist of $(G,k)$ where $G$ is a plane graph, and let $(G', k')$ be the corresponding reduced instance. According to Lemmas~\ref{lem_corr_elem}, \ref{lem_corr_1}, and \ref{lem_corr_2}, $G$ admits a \drb with size at most $k$ if and only if $G'$ admits a \drb with size at most $k' \leq k$. It is easy to see that the same time analysis of~\cite{AFN04} implies that our reduction rules can be exhaustively applied in time $O(|V(G)|^3)$. Let $D$ be a \drb of $G'$. Note that $|D| = 0 $ if and only if $G'$ is empty or has only one blue vertex, that is, $G'$ has constant size. Moreover, $|D| \neq 1$, since the unique dominating vertex should have been removed by \rrgl{rgl_som}. Also, $|D| \neq 2$, since the pair of dominating vertices should have been removed by \rrgl{rgl_pair}. Therefore, we may assume that $|D| \geq 3$, and then, according to Propositions \ref{prop_nb_reg}, \ref{prop_nb_excl}, and \ref{prop_nb_incl}, if $G'$ admits a \drb with size at most $k$, then $G'$ has order at most $14 \cdot (3k-6) + k \leq 43 k$. \end{proofof}

\section{Conclusion}
\label{sec:concl}

We have presented an explicit linear kernel for the \textsc{Planar Red-Blue Dominating Set} problem of size at most $43k$. A natural direction for further research is to improve the constant and the running time of our kernelization algorithm (we did not focus on optimizing the latter in this work), as well as proving lower bounds on the size of the kernel. It would also be interesting to extend our result to larger classes of sparse graphs. In particular, does \textsc{Red-Blue Dominating Set} fit into the recent framework introduced in~\cite{GPST13} for obtaining explicit and constructive linear kernels on sparse graph classes via dynamic programming?

A first step in this direction is a bikernel in the class of $H$-topological-minor-free graphs, which can be easily derived from the linear kernel for \dom in $H'$-topological-minor-free proved by Fomin \emph{et al.}~\cite{FLST13} combined with the following reduction from \RBDS to \dom proposed by an anonymous referee. Given an \RBDS instance $(G = (V_B \cup V_R,E), k)$, create a \dom instance $(G',k+1)$, where $G'$ is obtained from $G$ by adding a new vertex $u$ that is adjacent to all blue vertices, and to another new vertex $u'$ of degree 1. Given a \rbds $D$ of $G$, $D \cup\{u\}$ is a dominating set of $G'$. Conversely, given an optimal dominating set $D'$ of $G'$, the vertex $u'$ ensures that $u \in D$, thereby dominating all blue vertices. Hence to dominate $G'$ it suffices to dominate the red vertices (note that $D'$ does not contain red vertices because they only dominate themselves and blue vertices). The minor $H'$ is obtained from $H$ by adding a universal vertex. Such a bikernel is linear, but involves a large multiplicative constant depending on the excluded topological minor.

\vspace{.5cm}
\noindent \textbf{Acknowledgement}. We would like to thank the anonymous referees for helpful and thorough remarks that improved the presentation of the manuscript, and which allowed us to slightly improve the constant of our kernel. We also thank them for pointing out several imprecise steps in some of the proofs given in~\cite{AFN04} and for providing us helpful hints to fix them.

{\small
\bibliographystyle{abbrv}

}
\end{document}